\documentclass[12pt,leqno,draft]{amsart}
\usepackage{amsmath,amsthm,amssymb,verbatim,amsfonts}

\usepackage{mathrsfs}
\usepackage{mathtools}
\newtheoremstyle{fancy}{}{}{\itshape}{}{\textsc\bgroup}{.\egroup}{ }{}
\newtheoremstyle{fancy2}{}{}{\rm}{}{\textsc\bgroup}{.\egroup}{ }{}
\theoremstyle{fancy}
\newcounter{intro}

\numberwithin{equation}{section}    \swapnumbers

\newtheorem{lem}[equation]{Lemma}
\newtheorem{prop}[equation]{Proposition}
\newtheorem{thm}[equation]{Theorem}

\newtheorem{named}[equation]{\name}
\newcommand{\name}{Proof of}

\theoremstyle{fancy2}

\newcounter{axa}

\newtheorem{rem}[equation]{Remark}

\newcommand{\cref}[1]{Corollary~\ref{#1}}

\newcommand{\vol}{\operatorname{vol}}

\newcounter{subequation} 
 \newlength\mtabskip\mtabskip=-1.25cm
   
	 \def\mtabLong{long} 
 

\newcommand{\eq}[1]{\begin{equation}#1\end{equation}}

\newcommand{\alg}[1]{\begin{aligned}#1\end{aligned}}

\newcommand{\bc}{\begin{cases}}
\newcommand{\ec}{\end{cases}}

\newcommand{\Ab}[1]{\|#1\|}
\newcommand{\ab}[1]{|#1|}

\newcommand{\Abi}[1]{\|#1\|_{\infty}}

\newcommand{\al}{\alpha}

\newcommand{\la}{\lambda}

\newcommand{\lng}[1]{\langle{#1}\rangle}

\newcommand{\La}{\Lambda}
\newcommand{\om}{\omega}
\newcommand{\De}{\Delta}

\newcommand{\na}{\nabla}

\newcommand{\hp}{\hat{p}}

\newcommand{\mcr}[1]{\mathscr{#1}}

\begin{document}


\title[]{Area inequalities for stable marginally outer trapped surfaces 
in Einstein-Maxwell-dilaton theory}
\author{David Fajman, Walter Simon}
\address{
Gravitational Physics, Faculty of Physics\\
University of Vienna\\
Boltzmanngasse 5, 1090 Vienna \\
Austria}
\email{David.Fajman@univie.ac.at\\Walter.Simon@univie.ac.at}

\thanks{preprint UWThPh-2013-18}

\date{\today}


\begin{abstract}
We prove area inequalities for stable marginally outer trapped surfaces in
Einstein-Maxwell-dilaton theory. Our inspiration comes
on the one hand from a corresponding upper bound for the area 
in terms of the charges obtained recently by Dain, Jaramillo and Reiris \cite{DJR}
in the pure Einstein-Maxwell case without symmetries, and on the other hand
from  Yazadjiev's inequality \cite{SY} in the axially symmetric Einstein-Maxwell-dilaton case. 
The common issue in these proofs and in the present one is a functional
${\mathscr W}$ of the matter fields for which the stability condition readily
yields an {\it upper} bound. On the other hand, the step which crucially 
depends on whether or not a dilaton field is present is to obtain a {\it lower} 
bound for ${\mathscr W}$ as well. We obtain the latter 
by first setting up a variational principle for ${\mathscr W}$ with respect to the dilaton field
$\phi$, then by proving existence of a minimizer $\psi$ as solution 
of the corresponding Euler-Lagrange  equations 
and finally by estimating ${\mathscr W}(\psi)$. 
In the special case that the normal components of the 
electric and magnetic fields are proportional we obtain
the area bound $A \ge 8\pi P Q$ in terms of the electric and magnetic
charges. In the generic case our results are less explicit
but imply  rigorous `perturbation' results for the above inequality.
 All our inequalities are saturated for a 2-parameter family of static, 
extreme solutions found by Gibbons \cite{GA}.
Via the Bekenstein-Hawking relation $A = 4S$ our results give positive 
lower bounds for the  entropy $S$ which are particularly interesting in the 
Einstein-Maxwell-dilaton case.

\end{abstract}

\maketitle


\section{Introduction}

\subsubsection*{Geometric inequalities for MOTS}

Marginally outer trapped surfaces (MOTS) are compact 2-surfaces in spacetime on which the 
orthogonally outgoing family of null geodesics has vanishing expansion. 
In applications, `stable' and `strictly stable' MOTS play an important role and have been studied
 thoroughly \cite{RN}-\cite{JJ1}. Here stability is a
mild restriction which in essence requires that a certain 3-dim.\/\,neighbourhood of the MOTS 
can be foliated by trapped surfaces inside and untrapped surfaces outside.
While the stability condition came up first in connection with the topology 
theorems for MOTS \cite{RN, GG}, it also plays a key role in problems concerning 
the time evolution of initial data containing MOTS \cite{AMMS}, 
in particular in some versions of the singularity theorems \cite{IC}.

Recently there have been proven some remarkable area inequalities for stable
MOTS \cite{DJR}, \cite{SY}, \cite{SD1}-\cite{GJR}. They provide lower bounds for the area $A$ in terms of other
naturally available  parameters like electric and magnetic charges $P,Q$, 
the cosmological constant $\Lambda$ and  the angular momentum $J$ in the axially symmetric
case. The typical general form is $A \ge A_0(P,Q, \Lambda, J)$ for some
constant $A_0$ (but for $\Lambda > 0$ there are in addition upper bounds on
$A$, cf. Sect. 4.). These inequalities are purely quasilocal in the sense that they involve 
only the geometry in a neighborhood of the MOTS. In particular they are
$\it a ~priori$ not related to the more familiar Penrose inequalities 
(and conjectures) \cite{MM} which bound the ADM-mass $M$ from {\it below} in terms of 
the area of a MOTS and other quantities, viz.\/\,$M \ge M_0(A, P, Q, \Lambda, J)$.
 On the other hand, the area inequalities obviously provide some sort of `missing link' between Penrose-type
inequalities and `improved positive mass theorems' (`Bogomolny bounds'; cf.\,e.g.\,\cite{BC}-\cite{SZ})  by which we mean inequalities of the form $M\ge M_0'(P,Q, \Lambda, J)$, 
which are valid (or conjectured) irrespective of the presence of any MOTS (see also \cite{DKWY}).     

\subsubsection*{ MOTS in EMD spacetimes}
The derivation of the `basic' version of the area inequality, 
which needs suitable matter (-parameters) but no symmetries, is straightforward in principle.
If the MOTS has spherical topology (which is guaranteed if the energy-momentum tensor
$T_{\alpha\beta}$ satisfies the dominant energy condition) the stability condition implies 
(cf. Lemma (\ref{L1}))
\begin{equation}
\label{stb}
 \int_{\mathscr S} T_{\alpha\beta} k^{\alpha} \ell^\beta dS \le \frac{1}{2}.
\end{equation}
Here $k^{\alpha}$ and $\ell^\beta$ are future oriented null vectors orthogonal to the MOTS
scaled such that $\ell^\alpha k_{\al} = -1$, and our units are such that
Einstein's equations, which have been used to get (\ref{stb}), read $G_{\alpha\beta} = 8 \pi T_{\alpha\beta}$.
Now the remaining task consists of estimating the `surface energy' 
(\ref{stb}) from below in terms of the available parameters. In this process, the area 
either comes up automatically or can be brought in naturally.

In the present work we consider the Einstein-Maxwell-dilaton (EMD) system given
by the Lagrangian    
\begin{equation}
\label{emd}
\mcr L =  \left( R - 2 g^{\alpha\beta} \phi_{\alpha}\phi_{\beta} - e^{2 c \phi}
F_{\alpha\beta}F^{\alpha\beta} \right)
\end{equation}

where $\phi$ is the dilaton, $\phi_{\alpha} = \nabla_{\alpha} \phi$ its
gradient, $F_{\alpha\beta} = 2 \nabla_{[\alpha} A_{\beta]}$ is the electromagnetic field tensor 
with vector potential $A_{\beta}$ and $c$ a coupling constant.
For this system we can reformulate (\ref{stb}) as (cf.\/\,Lemma (\ref{L2}))
\begin{equation}
\label{Wphi}
 {\mathscr W(\phi)} =   \int_{\mathscr S} \left( \ab{D\phi}^2 + e^{-2c \phi} q^2   + e^{2c \phi} p^2
\right) dS  =
8\pi   \int_{\mathscr S} T_{\alpha\beta} k^{\alpha}\ell^\beta \le 4 \pi 
\end{equation}
where $D$ is the intrinsic derivative on ${\mathscr S}$. Now ${\mathscr W}(\phi)$ needs to be 
estimated from below. In the Einstein-Maxwell (EM) case ($\phi = 0$) this gives
\cite{DJR}
\begin{equation}
\label{aem}
A \ge 4\pi \left(Q^2 + P^2 \right)
\end{equation}
and some stronger estimate if a cosmological term is considered as well \cite{WS}.

\subsubsection*{Main results of the present work}
In the presence of the dilaton field our strategy differs  substantially from the
EM case. We start with minimizing $\mcr W(\phi)$ in (\ref{Wphi}) with respect to $\phi$,
and show existence of smooth solutions to the corresponding Euler-Lagrange 
equations by standard methods, namely via sub- and supersolutions.  Depending on how one
chooses these 
latter solutions the subsequent steps are different. 
We first (in Sect.\/\,3.1.3) consider the special case where the normal
components of the electric and magnetic fields, denoted by $q$ and $p$, are proportional on ${\mathscr S}$,
in which case we obtain the bound 
\begin{equation}
\label{apq}
A \ge 8\pi |P Q|.
\end{equation} 
The latter estimate is in fact saturated for a 2-parameter family of static, extreme black holes with coupling $c=1$ 
which are members of a 3-parameter family of static, generic solutions found by Gibbons \cite{GA}.
As to  the case with arbitrary $p$ and $q$, we first derive a rather crude bound (Thm.\,(\ref{T1})) which has the advantage of involving only 
the charges $Q$ and $P$ as well as $\sup_{\mathscr S}|q/p|$ and $\inf_{\mathscr S}|q/p|$. 
We then proceed  with a rather sophisticated estimate (Prop.\,(\ref{P1}))  which involves some global 
geometric parameters of $\mcr S$. This inequality implies in particular a quantitative stability statement 
for (\ref{apq}) (Thm.\,(\ref{T2})) in the sense that if
$|q/p|$ is almost constant, the resulting inequalities are close to (\ref{apq}).    

\subsubsection*{The axially symmetric case}
While our reasoning holds without assuming any symmetries, it certainly applies to the 
axially symmetric setting. However, in this case the natural and more ambitious goal is to obtain 
inequalities which explicitly contain the (Komar-) angular momentum $J$ on the r.h.s., and which are still saturated in non-trivial examples. 
For this task, the basic stability inequality (\ref{stb}) is insufficient and has to be replaced by a 
suitably adapted  variational principle. In the EM case  \cite{GJR}
obtained
\begin{equation}
\label{DJGR}
A \ge 4\pi \sqrt{4 J^2 + \left(Q^2 + P^2 \right)}  
\end{equation}
with equality in the extreme Kerr-Newman case.    
(This inequality was proven before in the stationary case \cite{HCA}). 
Similarly, for the EMD system (\ref{emd}) 
Yazadjiev  \cite{SY} showed that
\begin{equation}
\label{SY}
A \ge 8\pi \sqrt{ |J^2 - Q^2 P^2 |}
\end{equation}
which is saturated for the extreme, stationary Kaluza-Klein black holes (coupling $c= \sqrt{3}$)
and for the static ($J=0$) Gibbons solutions mentioned above. This latter example
therefore provides some connection with our results.

\subsubsection*{MOTS thermodynamics} A physical significance of the area inequalities comes from the
Bekenstein-Hawking relation $A = 4 S$ for the entropy $S$, (provided its
validity extends to MOTS, cf e.g.\,\cite{CS,AN}). In any case, the EMD case is of special interest
in black hole thermodynamics. In particular, in the extreme limit of that 
subfamily of the generic Gibbons solutions for which either $Q$ or $P$ vanishes,
one approaches a `black hole of zero area' - in fact a singularity - 
while the surface gravity stays finite. The (sloppy) thermodynamic interpretation of this 
object is `a black hole of zero entropy with non-zero temperature',
which has sparked some (serious) discussion (see e.g.\,\cite{HWIL,HH,GH}).
Our results are relevant for this issue in the sense that they seem to provide strictly positive 
lower bounds for the entropy. This will be discussed in Sect.\,5.   

\subsection*{Acknowledgement} We are grateful to Piotr Chru\'sciel for helpful
comments. W.S. was funded by the Austrian Science Fund (FWF): P 23337-N16.

\section{Preliminaries}

\subsection{Stable MOTS}
Here we consider a smooth, compact, connected spacelike 2-surface ${\mathscr S}$
in a spacetime $({\mathscr M},g_{\alpha\beta})$. 
Let  $k^{\alpha}$ and $\ell^\beta$ be future directed null vectors defined in a neighborhood of
${\mathscr S}$,  orthogonal to ${\mathscr S}$ on ${\mathscr S}$, and scaled such that  $k^{\alpha}
\ell_{\alpha} = -1$.  
The vector $\ell^\alpha$ is called \emph{outgoing}; its null expansion is defined as
\begin{equation}
\theta = h^{\alpha\beta} \nabla_{\alpha} \ell_{\beta},
\end{equation} 
where $h^{\alpha\beta}  =
(g^{\alpha\beta} + 2 k^{(\alpha} \ell^{\beta)})$ is the (inverse) 
metric on ${\mathscr S}$, 
and a MOTS is a surface ${\mathscr S}$ with $\theta \equiv 0$.
To define a stable MOTS with respect to a normal direction $m^{\alpha}$
\cite{RN,AMS2} we consider a foliation of a 3-neighborhood of ${\mathscr S}$ in the direction
$m^{\alpha}$ by a 1-parameter family of 2-surfaces ${\mathscr S}(\lambda)$, and we require that
the variation of $\theta$ in direction $m^{\alpha}$ is positive, viz. 
\begin{equation}
\delta_{\gamma m} \theta \ge 0,  
\end{equation}
where $\gamma > 0$ is the  lapse-function on ${\mathscr S}$ 
corresponding to the foliation ${\mathscr S}(\lambda)$.

We now recall the following Lemma \cite{RN,GG}.
\begin{lem}
\label{L1}
For a stable MOTS ${\mathscr S}$,
\begin{equation}
\label{varth}
0 \le  \int_{\mathscr S} \gamma^{-1} \delta_{\gamma m} \theta  \le
\int_{\mathscr S} \left[ \frac{1}{2}  R - 8\pi 
T_{\alpha\beta}\ell^\alpha k^{\beta} \right] dS = 
4\pi \left[(1 - g) - 2 \int_{\mathscr S}  T_{\alpha\beta}\ell^\alpha k^{\beta}  dS
\right]
\end{equation}
where $g$ is the genus of ${\mathscr S}$ and $R$ is the scalar curvature of ${\mathscr S}$.
\end{lem}
Here the final step of the proof uses the  Gauss-Bonnet theorem.

\subsection {Einstein-Maxwell-dilaton theory}

We consider the EMD Lagrangian of the form (\ref{emd})
with $\phi$, $A_{\alpha}$ and $g_{\alpha\beta}$ in the Sobolev space $H^3(\mcr M)$.
Our notation in essence follows Gibbons \cite{GA} who denotes the
coupling constant by $g$. We can assume $c > 0$ without loss of generality 
because of the symmetry ($c \rightarrow -c$, $\phi \rightarrow - \phi$)  
of the Lagrangian.
$c = \sqrt{3}$ is the Kaluza-Klein coupling \cite{GA,HWIS}.
 
The field equations obtained from $\mcr L$ (\ref{emd}) can be written as
\begin{eqnarray}
\Box \phi &  = & \frac{c}{2} e^{2c\phi} F^{\alpha\beta}F_{\alpha\beta} \\
\label{divF}
\nabla_{\alpha} \left( e^{2 c \phi} F^{\alpha\beta} \right)& = & 0 \\
R_{\alpha\beta} - \frac{1}{2} g_{\alpha\beta} R & = & 8\pi T_{\alpha\beta} 
\end{eqnarray} 
where
\begin{equation}
\label{emt}
8\pi T_{\alpha\beta} = 2 \left( \phi_{\alpha} \phi_{\beta} - \frac{1}{2}
g_{\alpha\beta} \phi_{\gamma} \phi^{\gamma} \right) + 2 e^{2c\phi} \left(
F_{\alpha}^{~\gamma} F_{\gamma\beta} - \frac{1}{4}
g_{\alpha\beta} F_{\gamma\delta} F^{\gamma\delta} \right)
\end{equation} 

We next define $q$ and $p$, the normal components on ${\mathscr S}$ of the electric and magnetic fields
by
\begin{equation}
\label{pq}
q = e^{2 c \phi} F_{\alpha\beta}
\ell^\alpha k^{\beta}  \qquad  p = *F_{\alpha\beta} \ell^\alpha k^{\beta} 
\end{equation}

\noindent and for the two-surface ${\mathscr S}$ we define the charges

\begin{equation}
\alg{
\label{PQ}
Q =  \frac{1}{4\pi} \int_{\mathscr S} q dS,  \qquad P = \frac{1}{4\pi} \int_{\mathscr S} p dS }.
\end{equation}
 By virtue of (\ref{divF}), $Q$ and $P$ are
independent of ${\mathscr S}$ which suggests (\ref{pq}) as 
reasonable definitions, even when as below only a fixed MOTS is considered. 
On the other hand, from the field equations there does not arise any natural definition of a dilaton 
charge of a 2-surface (see however Eq.\/\,(\ref{DC})). 
The regularity $\left(\phi,A_{\alpha}, g_{\alpha\beta}\right) \in H^3(\mcr M)$
assumed above guarantees that $q$ and $p$ are in $H^2(\mcr S)$ which will be used
in the proofs of the area inequalities in Sect.\,3.

\subsection{Static solutions}
We recall here first the 3-parameter Reissner-Nordstr\"om family 
in the EM case and then the 3-parameter family of static, spherically symmetric black hole
solutions for the coupling $c = 1$ found by Gibbons \cite{GA} which plays a canonical 
role in our estimates.

\subsubsection{Reissner-Nordstr\"om (coupling $c=0$)}
For ease of comparison with the Gibbons solution we write it in non-standard form, namely
\begin{eqnarray}
\eta & = & \frac{Q}{\rho + M} \qquad \chi = \frac{P}{\rho + M} \\
ds^2 &  = & - \lambda dt^2 + \lambda^{-1}\left[ d\rho^2 + \left( \rho^2 - L^2
\right) d\Omega^2 \right]
\end{eqnarray}
Here 
\begin{equation}
\lambda = \frac{\rho^2 - L^2}{(\rho + M)^2},
\end{equation}
the electric and magnetic potentials $\eta$ and $\chi$ are defined
from the vector potential $A_{\alpha}$ via 
\begin{equation}
\label{pots}
A_\alpha = (\eta, A_i), \qquad
B_i = \epsilon_{i}^{~jk} \partial_j A_k  \quad \mbox{and} \quad 
\lambda  B_i = - \partial_i \chi,
\end{equation}
and $M$ and $L^2 = M^2 - Q^2 -P^2$ are constants  with $L\ge 0$.
The coordinate $\rho$ is related to the usual radial coordinate
$r$ via $r = \rho + M$. The (outer) horizon ${\mathscr S}$ 
is located at $\rho =  L$; its  area and surface gravity are
\begin{equation}
\label{arn}
A = 4\pi \left( L + M \right)^2, \qquad 
\kappa = \frac{1}{2} \partial_{\rho}\lambda |_{\mathscr S} = 
\frac{L}{\left(L + M \right)^2}.
\end{equation}
The extreme solutions are characterized by $\kappa = 0 = L$; their 
area (\ref{arn}) is
\begin{equation}
A = 4\pi M^2 = 4\pi \left(Q^2 + P^2 \right)
\end{equation} 

\subsubsection{The Gibbons solution \cite{GA} (coupling $c = 1$). \label{gib}}

This family is given by
\begin{eqnarray}
e^{2\phi} = \frac{\rho + E}{\rho + F} \qquad \eta & = & \frac{Q}{\rho + E} \qquad \chi = \frac{P}{\rho + F} \\
ds^2 &  = & - \lambda dt^2 + \lambda^{-1}\left[ d\rho^2 + \left( \rho^2 - L^2
\right) d\Omega^2 \right]
\end{eqnarray}
Here 
\begin{equation}
\lambda = \frac{\rho^2 - L^2}{(\rho + E)(\rho + F)},
\end{equation}
the electric and magnetic potentials $\eta$ and $\chi$ are defined
as (\ref{pots}) except that now the dilaton enters in the definition of $\chi$
\begin{equation}
\label{pots}
A_\alpha = (\eta, A_i), \qquad
B_i = \epsilon_{i}^{~jk} \partial_j A_k  \quad \mbox{and} \quad 
\lambda e^{2 \phi} B_i = - \partial_i \chi,
\end{equation}
and $E$, $F$ and $L^2 = E^2 - 2 Q^2  = F^2 - 2P^2$ are constants  with $L\ge 0$.
The (outer) horizon ${\mathscr S}$ 
is again located at $\rho =  L$; its  area and surface gravity are
\begin{equation}
\label{gib}
A = 4\pi ( L + E)(L + F), \qquad 
\kappa = \frac{1}{2} \partial_{\rho}\lambda |_{\mathscr S} = 
\frac{L}{(L + E)(L + F)}.
\end{equation}
The extreme horizon is still characterized by $\kappa = 0 = L$ 
with the important proviso that we now have to assume $E = \sqrt{2} Q\neq 0$ and $F
= \sqrt{2}P \neq 0$
as the area 
\begin{equation}
A = 4\pi |E F|= 8\pi |P Q|.
\end{equation} 
would otherwise vanish. In fact, for $L = 0 = EB$ there is a lightlike singularity
at $\rho = L$ \cite{GH}, although the surface gravity $\kappa$ in (\ref{gib}) can be shown to remain 
finite in the extreme limit. This fact has stimulated some discussion, in particular regarding the 
thermodynamic interpretation \cite{HWIL,HH,GH}, which we recall in Sect.\,5.

\section{Area inequalities}

We first note that (\ref{emt}) satisfies the dominant energy condition,
whence any stable MOTS has spherical topology (Lem.\,(\ref{L1})). 
We continue with a Lemma which collects known results and will be key for what
follows. 

\newpage
\begin{lem}~
\label{L2}
\begin{enumerate}
\item
 For the Einstein-Maxwell-dilaton system (\ref{emd}) with 
$\left(\phi, A_{\alpha}, g_{\mu\nu} \right) \in H^3(\mcr M)$
the condition of stability of a MOTS (cf.\/\,Lem.\,(\ref{L1})) implies

\begin{equation}
\label{Wphipq}
 {\mathscr W(\phi,p,q)} =   \int_{\mathscr S} \left( \ab{D\phi}^2 + e^{-2c \phi} q^2   + e^{2c \phi} p^2
\right) dS  =
8\pi  \int_{\mathscr S} T_{\alpha\beta} k^{\alpha}\ell^\beta \le 4 \pi 
\end{equation}
where $D^{\alpha} = h^{\alpha \beta} \nabla_{\beta}$ is the intrinsic
derivative on ${\mathscr S}$.
\item Considered as a functional of $\phi$, a unique minimizer $\psi \in H^4(\mcr S)$ 
of ${\mathscr W}(\phi)$ always exists, satisfies the Euler-Lagrange
equations 
\begin{equation}
\label{psi}
\Delta \psi = c p^2 e^{2c\psi} - c q^2 e^{-2c\psi}
\end{equation}
and yields the estimate
\begin{equation}
\label{Wpsi} 
4 \pi \ge {\mathscr W}(\phi) \ge {\mathscr W}(\psi) = 
\int_{\mathscr S} \left[ (1 + c \psi) e^{-2c\psi} q^2 + (1 - c \psi) e^{2c\psi} p^2 \right] dS.
\end{equation} 

Moreover, if $\psi_{\pm} \in H^4(\mcr S)$ are sub- and supersolutions of (\ref{psi}),
i.e.
\begin{equation} 
\label{subsup}
\Delta \psi_- \ge c p^2 e^{2c\psi_-} - c q^2 e^{-2c\psi_-} \qquad \Delta
\psi_+ \le c p^2
e^{2c\psi_+} - c q^2 e^{-2c\psi_+}
\end{equation}
then $\psi_- \le \psi \le \psi_+$.

\end{enumerate}

\end{lem}

\begin{proof}[Proof.]~
\begin{enumerate}
 \item
This part follows by inserting (\ref{emt}) in Lemma \ref{L1} and 
treating the terms with the electromagnetic fields via Lemma 3.4.\,of \cite{DJR}. 

\item Existence of solutions of (\ref{psi}) and the bounds (\ref{subsup})
for $\psi$ follow from standard results
(Thm.\,1.10 and Prop.\,1.9 of \cite{MT}), while (\ref{Wpsi}) is obtained by partial integration.
\end{enumerate}

\end{proof}

Before proceeding with the general discussion of this equation we now discuss three special
cases. While the first two just use the first part of Lemma (\ref{L2})  and direct estimates
of (\ref{Wphi}),  the final one is non-trivial 
in the sense that it makes use of the second part of Lemma (\ref{L2}) as
well, in particular it uses an (albeit trivial) solution of (\ref{psi}).

\subsection{Special cases}

\subsubsection{Einstein-Maxwell}
In the  pure Einstein-Maxwell case ($\phi = 0$), the Cauchy-Schwarz estimates
 $\langle q^2 \rangle \ge 4\pi Q^2$ and $\langle p^2 \rangle \ge 4\pi P^2$
turn (\ref{Wphipq}) into  \cite{DJR} 
\begin{equation}
\label{nophi}
  16 \pi^2 \frac{Q^2 + P^2}{A} \le {\mathscr W}(\phi) \le 4\pi 
\end{equation}  
which gives (\ref{aem}); the bound is saturated for the extreme Reissner-Nordstr\"om solutions.

\subsubsection{Massless scalar field}
In the case $F_{\alpha\beta} = 0$ we obtain formally the same result 
\begin{equation}
\label{phi}
  16 \pi^2 \frac{Z^2}{A} \le {\mathscr W}(\phi) \le  4 \pi 
\end{equation}  
in terms of a `dilaton charge' 
\begin{equation}
\label{DC}
Z = \|D \phi\|_{L^1} = \frac{1}{4\pi} \int_{\mathscr S}  \ab{D \phi} dS.
\end{equation}
Note however that, in contrast to $Q$ and $P$ which are defined as integrals
over the {\it normal} components of the electric and magnetic fields,
$Z$ is defined from the {\it tangential} components $D\phi$.
Therefore $Z$  vanishes for all spherically symmetric solutions,
 no examples are known in which (\ref{phi}) is saturated,
while asymptotically flat static black holes do not exist anyway \cite{JC}.
Moreover, the integral is surface dependent, whence $Z$ will in general 
not  coincide with dilaton charges defined in the asymptotic region of 
asymptotically flat solutions. We anticipate that the quantity (\ref{DC})
will not appear in any of the estimates derived below.


\subsubsection{The case $p = \alpha q$ \label{paq}}
  
Here we assume that $p$ and $q$ are proportional with some
constant $\alpha \neq 0$.
As mentioned above, we now employ (\ref{psi}) which becomes   
\begin{equation}
\label{pap}
\Delta \psi = c  q^2 \left( \alpha^2 e^{2c\psi} -  e^{-2c\psi} \right).
\end{equation}
This is obviously solved by $e^{2c\psi} = |\alpha|^{-1}$,
and it is in fact the unique minimizer by virtue 
of Proposition 1.9 of \cite{MT} which is implicit in Lemma (\ref{L2}).
Inserting now in  (\ref{Wpsi}) and using as above the Cauchy-Schwarz estimate 
yields
\begin{equation}
\label{nDineq}
A \ge 8 \pi |P Q|.
\end{equation}
This bound is saturated for the extreme Gibbons solutions
(but not for extreme Reissner-Nordstr\"om), and consistent with Yazadjiev's 
inequality (\ref{SY}).

We note that uniqueness of  $e^{2c\psi} = |\alpha|^{-1}$, can also be seen directly 
rather than by using Lemma (\ref{L2}). Namely, integrating (\ref{pap}) over ${\mathscr S}$ 
implies that the expression in parenthesis, if nonzero, changes sign on ${\mathscr S}$. 
In particular it is positive at the maximum of $\psi$ and negative at
 a minimum of $\psi$. However, this contradicts the maximum principle and completes the argument.

\subsection{A sup-inf-estimate}

Here we assume that the normal components
$q$ and $p$ of the electric and magnetic fields do not vanish on ${\mathscr S}$, and we remark that this 
entails $Q\neq 0$ and $P \neq 0$. We define the quantities

\begin{equation}
\label{mupm}
e^{2c\mu_-} = \inf_{\mathscr S} \left|\frac{q}{p} \right|   \qquad
e^{2c\mu_+} = \sup_{\mathscr S} \left|\frac{q}{p} \right|   
\end{equation}
which play a key role in the following theorem.
Note that $\mu_- \le \mu_+$ since we assumed $c \ge 0$.  

 \begin{thm}
\label{T1}

 For the Einstein-Maxwell-dilaton system (\ref{emd}) with 
$\left(\phi, A_{\mu}, g_{\mu\nu} \right) \in H^3(\mcr M)$
we consider a smooth, stable MOTS $\mathscr S$ and assume that
$p \neq 0$ and $q \neq 0$ on ${\mathscr S}$.

Its area satisfies
\begin{equation}
\label{A1}
A \ge 4 \pi \left(N_q Q^2 + N_p P^2 \right)
\end{equation}
where 
\begin{eqnarray}
N_q & = & \min \left[\left(1 + c \mu_- \right) e^{-2 c \mu_-}, \left(1 + c
\mu_+  \right) e^{-2 c \mu_+} \right] \\
N_p & = & \min \left[\left(1 - c \mu_- \right) e^{2 c \mu_-}, \left(1 - c
\mu_+  \right) e^{2 c \mu_+} \right].
\end{eqnarray}
The bound (\ref{A1}) is saturated for the extreme Gibbons solutions (\ref{gib}).

\end{thm}

The proof will be based on Lemma (\ref{L1}), in particular on Eq.\,(\ref{Wpsi}), for which we need the following
lemma.

\begin{lem}
\label{L3}
The integrand ${ I }(\psi)$ of (\ref{Wpsi}), considered as function of $\psi$, 
has precisely one critical point in the interval $(-\infty, \infty)$ 
which is a maximum.
\end{lem}

\begin{proof}[Proof.] 
 A straightforward calculation shows 
\begin {eqnarray}
{I}'(\psi) & = & (1 - 2c\psi) p^2 e^{2c\psi} -  (1 + 2c\psi) q^2 e^{- 2c\psi}
\\
{I}'' (\psi) & = &  - 4c^2 \psi \left( p^2 e^{2c\psi}  - q^2 e^{- 2c\psi} \right)
\end{eqnarray}
and it is easy to see that ${I}'' < 0$ at critical points which are located at
\begin{equation}
\frac{p^2}{q^2} e^{4 c \psi} = \frac{1 + 2c\psi}{1 - 2c\psi}.
\end{equation}
A closer inspection shows that there is in fact precisely one such maximum.
\end{proof}

\begin{proof}[Proof of Theorem (\ref{T1})] 
We first note that the constants $\mu_-$ and $\mu_+$ defined in (\ref{mupm}) can serve as sub- and
supersolutions in Lemma (\ref{L2}). In fact we have
\begin{equation}
\Delta \mu_{-} = 0 \ge  c p^2 e^{2c\mu_-} - c q^2 e^{-2c\mu_-} \qquad 
\Delta \mu_{+} = 0 \le  c p^2 e^{2c\mu_+} - c q^2 e^{-2c\mu_+}  .
\end{equation}
Hence by Lemma (\ref{L2}) there is a  minimizer $\psi \in H^4(\mathscr S)$ 
which satisfies
$\mu_- \le \psi \le \mu_+$.

To estimate (\ref{Wpsi}) from below, we use Lemma (\ref{L3}) which implies
that, for any point of ${\mathscr S}$, the integrand ${I}(\psi)$ in (\ref{Wpsi}) takes its minimum 
for one of the boundary values $\mu_-$ or $\mu_+$. We now replace ${I}(\psi)$ by the minimal boundary value. A subtlety here is that 
it will in general depend on the point on ${\mathscr S}$ at which of the 
values $\mu_-$ or $\mu_+$ the minimum is taken. We therefore define
a partition of ${\mathscr S}$ in terms of two 
subsets ${\mathscr S}_-$ and ${\mathscr S}_+$ with 
\begin{equation}
{\mathscr S}_- \cap {\mathscr S}_+ = \emptyset \qquad 
{\mathscr S}_- \cup {\mathscr S}_+ =  {\mathscr S}
\end{equation}
such that on ${\mathscr S}_{\pm}$, ${I}$ takes its minimum
at $\mu_{\pm}$. Note that the definition of  ${\mathscr S}_{\pm}$ is not unique 
if there is a subset ${\mathscr S}_0\subset \mathscr S$ for which the minima 
${I}(\mu_-)$ and ${I}(\mu_+)$ coincide - 
in this case  ${\mathscr S}_0$ can be distributed arbitrarily among ${\mathscr
S}_{\pm}$. Inserting in (\ref{Wphi}) then yields
\begin{eqnarray}
 4 \pi
 \ge  {\mathscr W}(\psi) & = &\int_{\mathscr S} \left[ (1 + c \psi) e^{-2c\psi} q^2 + (1 - c \psi) e^{2c\psi} p^2 \right] dS
 \\
& \ge &  \int_{{\mathscr S}_-} \left[ (1 + c \mu_-) e^{-2c\mu_-} q^2 + (1 - c
 \mu_-) e^{2c\mu_-} p^2 \right] dS
 + \\
& + & \int_{{\mathscr S}_+} \left[ (1 + c \mu_+) e^{-2c\mu_+} q^2 + (1 - c \mu_+) e^{2c\mu_+} p^2 \right] dS
  \\
& \ge & 4\pi \min \left[\left(1 + c \mu_- \right) e^{-2 c \mu_-}, \left(1 + c \mu_+  \right) e^{-2 c \mu_+} \right]
\; \langle q^2 \rangle 
 + \\
& + & 4\pi \min \left[\left(1 - c \mu_- \right) e^{2 c \mu_-}, \left(1 - c \mu_+  \right) e^{2 c \mu_+} \right] 
\; \langle p^2 \rangle 
 \\
& \ge & 
 \frac{16 \pi^2}{A} \left(N_q Q^2 + N_p P^2 \right)
\end{eqnarray} 

where here and henceforth we use the notation 
\begin{equation}
\label{av}
\langle f \rangle = \frac{1}{4\pi} \int_{\mathscr S} f dS \qquad
\end{equation}
for the average of a quantity $f\in L^2({\mathscr S})$,
and we have as before used Cauchy-Schwarz estimates in the final step.

\end{proof}

An easy exercise is to verify that when $\mu_-$ approaches $\mu_+$, 
(\ref{A1}) tends to (\ref{nDineq}). A more concrete statement will 
be given in Theorem (\ref{T2}) in connection with the subsequent
estimate.

\subsection{An $L^2$-estimate}


\begin{prop}\label{P1}
 For the Einstein-Maxwell-dilaton system (\ref{emd}) with 
$\left(\phi, A_{\mu}, g_{\mu\nu} \right)$\\ $\in H^3(\mathscr M)$
we consider a smooth, stable MOTS $\mathscr S$. 

 We denote by 
$\la_1 = \la_1(\mathscr S)$ the first non-zero eigenvalue of the Laplace-Beltrami operator on 
$\mathscr S$, and we use the notation $\hat f = f/\sqrt{\langle f^2 \rangle}$ for a quantity $f\in L^2(\mcr S)$. 

Then there exists a function $\varphi\in H^2(\mathscr S)$ with 
\eq{\label{eq-abs}
\Abi{\varphi}\leq C(\la_1,\inf_{x\in\mathscr S} R) \sqrt{\langle p^2\rangle\langle
q^2\rangle}\sqrt{\lng{(\hat p^2 - \hat q^2)^2}} 
} 
such that
\eq{\alg{\label{West}
\mcr W&\geq \sqrt{\langle p^2\rangle\langle q^2\rangle}\Bigg[\int_{\mathscr S} e^{-2c\varphi}\hat q^2+e^{2c\varphi}\hat p^2dS\\
&\quad+\frac14\ln(\lng {q^2} / \lng {p^2})\int_{\mathscr S}\left(e^{-2c\varphi}\hat q^2-e^{2c\varphi}\hat p^2\right)dS\\
&\quad+ c\int_{\mathscr S}\varphi\left(e^{-2c\varphi}\hat q^2-e^{2c\varphi}\hat p^2\right)dS\bigg]
}} 
holds, where $C = C(\la_1,\inf_{x\in\mathscr S} R)$ is a constant depending on the geometry
of $\mathscr S$.
Moreover, if $\hat q = \hat p$ this reduces to ${\mathscr W} \ge 32 \pi^2 |P Q| /A$
which is saturated for the Gibbons solutions (\ref{gib}).

\end{prop}

A straightforward consequence of Proposition (\ref{P1}) is the following `perturbation'
of inequality (\ref{nDineq}). 

\begin{thm}\label{T2}
Under the assumptions of Proposition (\ref{P1}), for all $\varepsilon>0$ there exists a $\delta>0$ such that 
for $\langle (\hat p^2 - \hat q^2)^2\rangle <\delta$, 
\begin{equation}
\label{aepq}
 A\geq 8\pi  (1-\varepsilon) |P  Q|.
\end{equation}

\end{thm}
Note that in contrast to Theorem (\ref{T1}), $q$ and $p$ are allowed to have
zeros on ${\mathscr S}$ in Proposition (\ref{P1}) and Theorem (\ref{T2}).
For $\delta = 0$, we have $p = \alpha q$ and the result reduces to (\ref{paq}).

\begin{proof}[Proof of Proposition (\ref{P1})]
We set out from Lemma (\ref{L1}) but instead of taking constant sub- and supersolutions
as in Theorem (\ref{T1}), we follow the method used by Choquet-Bruhat and Moncrief in another context \cite{CM01}.
Consider the equation
\eq{
0=e^{2c\omega}\lng{p^2}-e^{-2c\omega}\lng{q^2},
}
for a real number $\omega$, yielding
\eq{
\lng{q^2}\lng{p^2}^{-1}=e^{4c\omega}.
}
Then the linear equation
\eq{\label{lin}
\De v=e^{2c\om}p^2-e^{-2c\om}q^2=f(\omega)
}
has a unique solution $v\in H^4({\mathscr S})$ since $\int_{\mathscr S}f(\omega)=0$. 
Sub- and supersolutions are now defined in terms of $v$ by
\begin{equation}
\label{psipm}
\psi_+=v-\min v+\om\qquad\mbox{and}\qquad \psi_-=v-\max v+\om
\end{equation}
which implies (\ref{subsup}). By Lemma (\ref{L1}), there exists a unique solution $\psi\in H^4(\mcr S)$ to (\ref{psi}), with 
\eq{
\mcr W(\zeta)\geq\mcr W(\psi)\qquad\forall \zeta\in H^{1}(\mcr S)
}
and 
\eq{\label{psibd}  \psi_-\leq\psi\leq\psi_+.
}
For $\varphi$ defined by
\eq{\varphi = \omega - \psi}
(\ref{psipm}) and (\ref{psibd})
imply the straightforward pointwise estimate
\eq{
\Ab{\varphi}_{L^{\infty}(\mathscr S)}\leq 2\Ab{v}_{L^{\infty}(\mathscr S)}.
}
To estimate the minimizer $\psi$  we will now estimate $v$ in terms of
geometric quantities.\\
As $v$ is a solution of (\ref{lin}), partial integration yields 
\eq{
\int\ab{\na v}^2dS=-\int v f(\om).
}
which in turn gives 
\eq{\label{es-1}
\Ab{\na v}_{L^2}^2\leq\Ab{f(\om)}_{L^2}\Ab{v}_{L^2}.
}
The Poincar\'e inequality for $v$ - as a function with mean value zero -
reads
\eq{\label{poin}
\Ab{v}_{L^2}^2  \leq\frac1{\la_1(g)}\Ab{\na v}_{L^2}^2.
}
With (\ref{es-1}) this now implies

\eq{\label{es-2}
\Ab{\na v}_{L^2}\leq\frac1{\sqrt{\la_1(g)}}\Ab{f(\om)}_{L^2}.
}
Finally, the Ricci identity
\eq{
\Ab{\na^2 z}_{L^2}^2=\Ab{\De z}_{L^2}^2-\frac12\int R\ab{\na z}^2, \qquad
\forall z\in H^2(S)
}
gives
\eq{\label{es-3}
\Ab{\na^2v}_{L^2}^2\leq\Ab{f(\om)}_{L^2}^2+\frac{\inf R}2\int \ab{\na
v}^2dS.
}
Combining (\ref{poin}), (\ref{es-2}) and (\ref{es-3}) yields the estimate
for $v$
\eq{\label{es-4}
\Ab{v}_{H^2}\leq \sqrt{1+\frac1{\la_1(g)^2}(1 + \la_1(g)) +\frac{\inf
R}{2\la_1(g)}}\Ab{f(\om)}_{L^2}.
}
We note that if the scalar curvature $R$ is non-negative on ${\mathscr S}$, the
term
with $\inf R$ can be dropped from  (\ref{es-3}) and hence from (\ref{es-4}).
Now Sobolev embedding with optimal constant $C(g)$ yields
\eq{
\Abi{v}\leq C(g) \sqrt{1+\frac1{\la_1(g)^2}(1 + \la_1(g)) +\frac{\inf
R}{2\la_1(g)}}\Ab{f(\om)}_{L^2},
}
which is (\ref{eq-abs}). Next, $f(\om)$ can be estimated as follows
\eq{\alg{
\Ab{f(\om)}^2_{L^2} &= 4c^2\int_{\mathscr S}\left(e^{2c\om}p^2-e^{-2c\om}q^2\right)^2dS\\
&=4c^2\lng{q^2}\lng{p^2}\lng{\left(\hat p^2-\hat q^2\right)^2}.
}}
Eventually, using (\ref{Wpsi})  and resubstituting $\psi=\varphi+\om$ yields
\eq{\alg{
\mcr W&\geq \int -2c(\omega+\varphi)(p^2e^{2c(\om+\varphi)}-q^2e^{-2c(\om+\varphi)})\\
&\qquad+p^2e^{2c(\om+\varphi)}+q^2e^{-2c(\om+\varphi)}dS\\
&=\sqrt{\lng{p^2}\lng{q^2}}\Bigg[\int  e^{-2c\varphi}\hat q^2 +e^{2c\varphi}\hp^2dS\\
&\qquad+\frac12\ln(\lng q^2/\lng p^2)\int(\hat q^2e^{-2c\varphi}-\hat p^2e^{2c\varphi})dS\\
&\qquad+2c\int \varphi(\hat q^2e^{-2c\varphi}-\hat p^2e^{2c\varphi})dS\bigg].
}}
which was claimed in (\ref{West})
\end{proof}

We recall that, compared to Yazadjiev's result (\ref{SY}), our theorem
does not require axial symmetry. In the axially symmetric case,
however, our estimate (\ref{aepq}) gives complementary infomation whose 
relevance clearly depends on the relative magnitude of the quantities $|P Q|$, $J^2$ and  
$\langle (\hat p^2 - \hat q^2)^2 \rangle$.

\section{The cosmological case}

A class of Lagrangians with cosmological constant 
with arbitrary dilaton coupling constant $d$, namely

\begin{equation}
\label{emdL}
\mcr L =  \left( R - 2 g^{\alpha\beta} \phi_{\alpha}\phi_{\beta} - e^{2 c \phi}
F_{\alpha\beta}F^{\alpha\beta} - \Lambda e^{-d \phi} \right)
\end{equation}
has been discussed recently (cf e.g.\,\cite{CLSZ}). It seems interesting to 
 generalize  the results for stable MOTS obtained above to this case, but 
this lies beyond the scope of the present paper. We rather restrict ourselves here
to the case $d = 0$ for which the generalization is straightforward. 
As  Einstein's equations now read
\eq{
G_{\alpha\beta} +\La g_{\alpha\beta} = 8\pi T_{\alpha\beta}
}
 the stability condition (\ref{Wphi}) has to be replaced by  
\begin{equation}
\label{WphiL}
 {\mathscr W(\phi)} =   \int_{\mathscr S} \left( \ab{D\phi}^2 + e^{-2c \phi} q^2   + e^{2c \phi} p^2
\right) dS  =
  \int_{\mathscr S} \left( 8\pi T_{\alpha\beta} - \Lambda g_{\alpha\beta} \right)k^{\alpha}\ell^\beta \le
 4\pi  + \Lambda A.  
\end{equation}

This leads to the following quadratic inequality for $A$ as 
generalization of Theorem (\ref{T1}) 

\begin{eqnarray}
\Lambda A^2 - 4\pi (1 - g) A + 16\pi^2 \left(N_q Q^2 + N_p P^2 \right) \le 0,
\end{eqnarray}
where $g$ is the genus of ${\mathscr S}$. This inequality can be discussed
along the lines of \cite{WS}; for $\Lambda > 0$ there  arise an upper and a lower bound for $A$, as well as the upper
bound  $ N_q Q^2 + N_p P^2 \le 1/(4\Lambda)$. On the other hand, for
$\Lambda < 0$, there is just a lower bound for the area.   

In a similar manner Theorem (\ref{T2}) now reads that 
for all $\varepsilon>0$ there exists a $\delta>0$ such that 
for $\langle (\hat p^2 - \hat q^2)^2\rangle < \delta,$ 

\begin{equation}
\Lambda A^2 - 4\pi (1 - g) A + 32\pi^2 (1 - \varepsilon) |P Q| \le 0
\end{equation}
with area bounds as before; for $\Lambda > 0$ the upper bound on the charges
now reads  $(1 - \varepsilon)|P Q| \le 1/(2 \Lambda)$.



\section{Black holes versus elementary particles}

As indicated in the introduction and in Sect.\,2.3.2., 
static EMD black holes have interesting thermodynamic properties in the extreme limit.  
We first recall the  cases in which solutions are known explicitly.
\subsubsection{Arbitrary $c$ but either electric or magnetic field vanishing
\label{gm}}
In this case there is the  2-parameter family of solutions found in
\cite{GM} which contains an extreme 1-parameter
subfamily. The area of the horizons goes to zero in the extreme
limit and for all $c \neq 0$. More precisely, this family approaches a
singular solution where the singularity was found to be lightlike
\cite{GH}. As to the surface gravity $\kappa$, 
 the value $c = 1$ is critical in the sense that for $c < 1$, $\kappa$ goes
 to zero in the extreme limit, for $c=1$ it approaches a constant, while for 
$c > 1$ it diverges.

\subsubsection{Coupling $c=0$ or $c=1$, and arbitrary Maxwell field}
Here there are the 3-parameter families of Reissner-Nordstr\"om and 
Gibbons solutions with 2-parameter extreme subfamilies, which have been reviewed in Sect.\,2.3.2.
For the Gibbons solutions ($c=1$) we have seen that the
area satisfies $A = 8\pi |P Q|$; in particular, it stays positive in the
extreme limit provided that neither charge vanishes.   \\
\\
Of course the thermodynamic interpretation of the extreme limit is subtle 
when the zero area $\widehat{=}$ zero entropy $\widehat{=}$  singular solution with
non-zero surface gravity $\widehat{=}$ temperature is approached. 
Based on perturbation analyses with axially symmetric, time dependent
perturbations and with a spherically symmetric, time dependent test field,
 Holzhey and Wilczek \cite{HWIL} have argued that for $c = 1$ the static extremal solutions
(\ref{gm}) enjoy a finite `mass gap' in the sense  that these objects are repulsive for 
low-energy perturbations, while for $ c > 1$ the mass gap is infinite and the
object is universally repulsive. Accordingly, these authors put forward an analogy
between these configurations and elementary particles.

In the static, spherically symmetric examples discussed above,  
the 'particle-like' behaviour is associated with singular solutions having 
'zero horizon area' and therefore zero entropy. 
If this connection persists in non-symmetric, dynamical
spacetimes  (in the dynamical case, the proportionality between entropy and area has been 
under dispute \cite{CS,AN}), Yazadjiev's \cite{SY} and our theorems provide trivial
criteria to judge potential 'particle candidates': The former results, which
apply to stable axially symmetric MOTS, single out configurations with $|PQ| = |J|$
as good candidates since  $A = 0$,  
while they exclude those with  $|PQ| \neq |J|$ for which $A > 0$. 
On the other hand, in the generic case our theorems (\ref{T1}) and (\ref{T2})
exclude several  configurations with stable MOTS, again by guaranteeing
that $A > 0$. In particular, theorem (\ref{T2}) excludes solutions for which  
$P.Q \neq 0$ and for which  $q$ and $p$ 
are either proportional or close to being proportional.


\end{document}